\documentclass[]{llncs}
    
\usepackage[english]{babel}
\usepackage{xurl}
\usepackage{mdframed}
\usepackage{comment}
\usepackage{color}
\usepackage{amsfonts}
\usepackage{algorithmic}
\usepackage{algorithm}
\usepackage{array}
\usepackage[caption=false,font=normalsize,labelfont=sf,textfont=sf]{subfig}
\usepackage{textcomp}
\usepackage{stfloats}
\usepackage{hyperref}
\usepackage{verbatim}
\usepackage{graphicx}

\newcommand{\ignore}[1]{}

\newcommand{\calZ}{\mathcal{Z}}

\newcommand{\bin}{\{0,1\}}
\newcommand{\secpar}{\lambda}
\newcommand{\sgn}{\sigma}
\newcommand{\set}{\Sigma}
\newcommand{\pk}{\mathit{pk}}
\newcommand{\sk}{\mathit{sk}}

\newcommand{\alg}[1]{\mathsf{#1}}

\newcommand{\Psc}{\Pi_{SC}}
\newcommand{\Paf}{\Pi_{FP}}
\newcommand{\Reg}{\mathit{Reg}}
\newcommand{\DevIss}{\mathit{D}}
\newcommand{\smartct}{\mathit{C_{SC}}}

\newcommand{\env}{\calZ}

\newcommand{\party}{\alg{P}}
\newcommand{\advA}{\alg{A}}
\newcommand{\advS}{\alg{S}}
\newcommand{\advD}{\alg{Z}}
\newcommand{\advC}{\mathsf{Chall}}

\newcommand{\initialized}{\mathtt{init}}
\newcommand{\init}{\mathtt{init}}

\newcommand{\fieldid}{\mathtt{areaId}}
\newcommand{\devid}{\mathtt{deviceId}}
\newcommand{\itemid}{\mathtt{itemId}}
\newcommand{\assetid}{\mathtt{assetId}}
\newcommand{\batchid}{\mathtt{batchId}}
\newcommand{\categid}{\mathtt{categId}}

\newcommand{\auddata}{\mathtt{auditData}}
\newcommand{\itm}{\mathtt{item}}
\newcommand{\items}{\mathtt{items}}
\newcommand{\products}{\mathtt{products}}
\newcommand{\p}{\mathtt{p}}
\newcommand{\newstate}{\mathtt{newstate}}

\newcommand{\tEvent}{\mathtt{tEvent}}
\newcommand{\aggregation}{\mathtt{aggregation}}
\newcommand{\disaggregation}{\mathtt{disaggregation}}
\newcommand{\cf}{\mathtt{categFP}}
\newcommand{\res}{\mathtt{result}}
\newcommand{\tx}{\mathtt{tx}}

\newcommand{\intact}{\mathtt{intact}}
\newcommand{\packaged}{\mathtt{packaged}}
\newcommand{\trans}{\mathtt{trans}}
\newcommand{\destroyed}{\mathtt{destroyed}}
\newcommand{\submit}{\mathtt{submit}}
\newcommand{\rd}{\mathtt{read}}

\newcommand{\reg}{\mathtt{register}}
\newcommand{\regok}{\mathtt{registered}}
\newcommand{\attr}{\mathtt{attr}}

\newcommand{\recvok}{\mathtt{received}}

\newcommand{\rejok}{\mathtt{rejected}}
\newcommand{\creat}{\mathtt{create}}

\newcommand{\transform}{\mathtt{transform}}
\newcommand{\training}{\mathtt{training}}

\newcommand{\merge}{\mathtt{merge}}

\newcommand{\dev}{\mathtt{device}}

\newcommand{\farm}{\mathtt{produce}}
\newcommand{\scsplit}{\mathtt{split}}
\newcommand{\upd}{\mathtt{update}}
\newcommand{\hand}{\mathtt{handover}}
\newcommand{\farmi}{\mathtt{producing}}
\newcommand{\creati}{\mathtt{creation}}
\newcommand{\transi}{\mathtt{transformation}}

\newcommand{\train}{\mathtt{training}}

\newcommand{\audit}{\mathtt{audit}}

\newcommand{\regsc}{\mathtt{register_{SC}}}
\newcommand{\regscok}{\mathtt{registered_{SC}}}
\newcommand{\farmsc}{\mathtt{produce_{SC}}}

\newcommand{\farmiscok}{\mathtt{produced_{SC}}}
\newcommand{\creatsc}{\mathtt{create_{SC}}}
\newcommand{\creatscok}{\mathtt{created_{SC}}}
\newcommand{\transsc}{\mathtt{transform_{SC}}}
\newcommand{\transscok}{\mathtt{transformed_{SC}}}
\newcommand{\trainsc}{\mathtt{training_{SC}}}

\newcommand{\trainscok}{\mathtt{trained_{SC}}}
\newcommand{\audsc}{\mathtt{audit_{SC}}}

\newcommand{\audscenabl}{\mathtt{enableaudit_{SC}}}
\newcommand{\mergsc}{\mathtt{merge_{SC}}}
\newcommand{\mergscok}{\mathtt{merged_{SC}}}
\newcommand{\splitsc}{\mathtt{split_{SC}}}
\newcommand{\splitscok}{\mathtt{splitted_{SC}}}
\newcommand{\handsc}{\mathtt{handover_{SC}}}
\newcommand{\handscok}{\mathtt{handoverEnded}}
\newcommand{\handscko}{\mathtt{handoverFailed}}
\newcommand{\recvsc}{\mathtt{receive_{SC}}}
\newcommand{\recvscok}{\mathtt{received_{SC}}}
\newcommand{\rejsc}{\mathtt{reject_{SC}}}
\newcommand{\rejscok}{\mathtt{rejected_{SC}}}
\newcommand{\updsc}{\mathtt{update_{SC}}}
\newcommand{\updscok}{\mathtt{updated_{SC}}}
\newcommand{\readsc}{\mathtt{read_{SC}}}

\newcommand{\initff}{\mathtt{init_{FF}}}

\newcommand{\initffok}{\mathtt{initialized_{FF}}}
\newcommand{\trainff}{\mathtt{train_{FF}}}
\newcommand{\trainffko}{\mathtt{trainingfailed_{FF}}}
\newcommand{\updffok}{\mathtt{updated_{FF}}}
\newcommand{\evalff}{\mathtt{eval_{FF}}}
\newcommand{\evalffko}{\mathtt{evalfailed_{FF}}}
\newcommand{\evalffok}{\mathtt{evalued_{FF}}}
\newcommand{\handff}{\mathtt{handover_{FF}}}

\newcommand{\recvffok}{\mathtt{received_{FF}}}

\newcommand{\readff}{\mathtt{read_{FF}}}

\newcommand{\deldev}{\mathtt{deleteDevice}}
\newcommand{\deldevok}{\mathtt{deletedDevice}}

\newcommand{\submitll}{\mathtt{submit_{L}}}
\newcommand{\readll}{\mathtt{read_{L}}}

\newcommand{\sendsmt}{\mathtt{send_{SMT}}}
\newcommand{\sendsmtok}{\mathtt{sent_{SMT}}}

\newcommand{\Func}{\ensuremath{\mathcal{F}}}
\newcommand{\Gfunc}{\ensuremath{\mathcal{G}}}
\newcommand{\Fsc}{\Func_{SC}}
\newcommand{\Fledg}{\Func_{Ledger}}
\newcommand{\Fsig}{\Func_{Sig}}
\newcommand{\Fsmt}{\Func_{SMT}}
\newcommand{\Fff}{\Func_{FP}}

\newcommand{\AFS}{FP scanner}
\newcommand{\AFSs}{\AFS s}

\newcommand{\rv}[1]{\mathbf{#1}}

\newcommand{\REAL}{\rv{REAL}}
\newcommand{\IDEAL}{\rv{IDEAL}}

\newcommand{\HYB}{\rv{H}}
\newcommand{\NN}{\mathbb{N}}
\newcommand{\cind}{\approx_c}

\begin{document}
\title{Secure Blockchain-Based Supply Chain Management with Verifiable Digital Twins\thanks{Part of the contributions presented in this work appeared in the paper ``Secure Blockchain-Based Supply Chain
Management with Verifiable Digital Twins'' published in proceedings of ``ACM 3rd International Conference on Information Technology for Social Good (GoodIT 2023)'',  pages 291–298, DOI: 10.1145/3582515.3609547 \cite{BFMV23}.}}

\author{Vincenzo Botta\inst{1} \and
Laura Fusco \inst{2} \and
Attilio Mondelli\inst{3}\and
Ivan Visconti\inst{2}
}

\institute{University of Warsaw, Warsaw, Poland \email{v.botta@uw.edu.pl} \and DIEM, University of Salerno, Italy \email{laurafusco1995@gmail.com,visconti@unisa.it}\and Farzati Tech., Italy \email{attilio.mondelli@farzatitech.it}}

\maketitle
\begin{abstract}
A major problem in blockchain-based supply chain management is the
potential unreliability of digital twins when considering digital representations
of physical goods. Indeed, the use of blockchain technology to trace goods is obviously ineffective
if there is no strong correspondence between what is physically exchanged
and the digital information that appears  
in blockchain transactions.

In this work, we propose a model for strengthening
the supply chain management of
physical goods by leveraging blockchain technology along with
a digital-twin verification feature. Our model can be instantiated
in various scenarios and we have in particular considered the
popular case of food traceability.  

In contrast to other models known in the literature that 
propose their own ad-hoc properties to assess
the robustness of their supply chain management systems,
in this work we use the formalism of secure computation, where processes
are described through generic and natural ideal functionalities.
\end{abstract}


\section{Introduction}\label{sec:intro}

The traceability of goods is extremely important to optimize their management
and to reduce risks of frauds and their consequences 
(e.g., in foodstuffs frauds damage economy and public health).
Mutually distrustful 
supply chain (SC) members are interested in using tracing systems 
since traceability provides transparency and helps to fight counterfeiting 
with positive effects on the reputation and business of the involved (honest) companies~\cite{VSMWV12}.

Traditional supply chain management systems are 
centralized and their credibility relies on the trust towards central administrators \cite{Tia16,ATM17}. This single point of failure due to centralized architectures is a classical issue that motivates the use of blockchain technology. 
Indeed, blockchains can help providing an immutable, consistent, transparent, decentralized and highly available ledger of supply chain data. 
While originally focusing on cryptocurrencies, blockchain technology is
nowadays mature enough to be used in other scenarios 
decentralizing classic centralized solutions.
 Leveraging on blockchain technology all events in a supply chain
can be digitized and stored in a transparent and immutable manner with robust
data availability. Indeed, current blockchain-based solutions provide the traceability of the supply chain 
by recording the status of a good at each stage of the supply chain. 
\paragraph{Achilles' heel: the validity of a digital twin.}
A supply chain management system provides ways to connect a virtual digital good to a concrete physical good~\cite{MDR21}.
Unfortunately, there is a major issue
to consider:
the problematic connection between 
data stored in a blockchain and the actual goods in the physical world \cite{PHJAB20}.
Consider for instance a farm producing mozzarella with milk obtained from Italian Mediterranean buffalo. 
A blockchain is useful to transparently give information on the location and size
of the farm, the number of buffalos and the amounts of buffalo mozzarella that is daily produced and
reach the stores  and the final consumers. The use of a blockchain allows to publicly verify that the farm does not 
claim unreasonably high amounts of produced buffalo mozzarella. While this tracing is beneficial for the entire system,
it still remains possible to dishonestly mix buffalo milk and cow milk whenever there is shortage of buffalo milk, due for instance to a virus affecting buffaloes. This is a typical case where while data appearing in the blockchain are consistent with expectations,
they can be completely uncorrelated to the physical goods to which they refer to. 

\subsection{Our Contribution}
In this paper, we present a general model for blockchain-based supply chain focusing on the traceability and the verification of goods. 

Our goal is to enable every party in the system
to retrieve full information about the life-cycle
of goods in the supply chain. Each 
physical good has a digital twin identified by a unique code 
that appears on a blockchain along with all the relevant events describing the life of the good (e.g., the event that generated it, the merging and splitting as a batch). A reliable auditing process is required in order
to strengthen consistency between the information stored on-chain and the physical goods, therefore reducing errors and frauds. Every time an inspection is carried out by a supply chain entity some relevant information (e.g.,  
date, time, place, operator, verified goods, the outcome of the inspection)
about the inspected good is recorded on the blockchain. 

We formalize the concept of reliable audit through a specific ideal functionality that can be instantiated in several
ways. In particular our model focuses on establishing how to declare a category of goods that will be involved in the supply chain
and then how to make sure that goods that are supposed to belong to a certain category can be tested in case of inspection.
Notice that defining this process is non-trivial given the evolving nature of goods during their life-cycles (e.g., the color,
hardness and flavor of a banana are expected to change day by day depending also on how they are stored).

When presenting our model we will be generic about the goods, but for concreteness we will sometimes
refer in particular to food traceability. Indeed, this use case
represents a popular example where the
use of blockchains has been widely considered and  bogus digital twins can severely affect
desired advantages provided by the publicly verifiable tracing.  

The formalism that we use to guarantee the robustness of the traceability provided by our model 
is the gold standard in the design of secure systems: secure multi-party computation (SMPC)~\cite{GMW87}. 
It consists of defining a natural ideal world representing how traceability can be reliably performed in the presence of trusted parties. Then, the reliability of 
a concrete real-world system can be measured by showing that any attack on the concrete system can be (through a simulation) carried also in the above ideal world. Since by inspection one can easily check the reliability of processes run in the ideal world, and since the real-world system is proven as secure as executions in the ideal world, one can conclude that the proposed  real-world system is reliable too. 

\section{Related Work}\label{sec:related_supply}
Our work mainly focuses on designing a framework that combines a permissioned blockchain together with a reliable auditing process that together makes possible the traceability and the verification of goods in a supply chain.

The concept of the digital twin was discussed by Grieves in \cite{gri14}. The main idea introduced by Grieves is to create an equivalent
digital representation of a physical object in
real life and to link the physical object with its representation. 
Grieves described the digital twin as composed of three
components: a physical product, a virtual representation of that
product, and the bi-directional data that connect the physical product and the virtual representation.
From this pioneering work, many researchers conducted studies to design digital twin systems that try to guarantee an equivalent digital representation of a manufacturing object. 

Many works have summarized and reviewed the characteristics of digital twin technologies and their application developed in the last decades.  Jones {\em et al.} in \cite{JONES202036} have outlined the identity core of the digital twin applications through a literature review. 
Deng {\em et al.} \cite{DENG2021125} reviewed works on digital twin applications in urban governance. 
Marmolejo-Suacedo {\em et al.} \cite{SaucedoHS19} presented a literature review of the concepts of digital twins in the supply chain and their limitations.
Defraeye {\em et al.} in \cite{DEFRAEYE2021245} identified key advantages of digital twins that can be exploited by the supply chain in the context of fresh horticultural products.
Rivera {\em et al.} in \cite{RMVT020} propose  GEMINIS, a model to specify the structural and behavioural aspects of digital twins. Yaqoob {\em et al.} in \cite{YSUJOI20} analyze the benefits of employing blockchain in systems to manage digital twins. However, their work lacks of formal analysis to assess the security of their application. Guo {\em et al.} in \cite{GLLAZRH20} analyze the problem of managing personalized services for products with a limited life-cycle. They focus on the possibility of designing a system that, with the aim of the blockchain, produces personalized manufacturing in which customers can participate and are involved in improving production. Lee {\em et al.} in \cite{LEE21} focus on accountable information sharing among participants in the digital twin system. The authors in particular highlight that the information in the digital twin must be tamper-proof and available in a traceable way to all participants. Their work however does not rigorously model the  auditing process that should guarantee the links between the digital and physical world.
In \cite{TVB21} authors present preliminary results for digital twins in the context of the food supply chain to assess the state-of-the-art of this technology.
Jabbar {\em et al.} in \cite{JLHAR21} proposed a survey to cover many aspects of supply chain management with blockchains. They analyze the technical and non-technical challenges to adopt blockchain for supply chain and propose a model called MOHBSChain, for the adoption of blockchains. Jabbar {\em et al.} identify blockchain as a promising technology for addressing secure traceability, data immutability and trust building. According to Jabbar {\em et al.} many companies like IBM\footnote{See \url{https://www.ibm.com/blockchain/supply-chain}.}, and Walmart\footnote{See \url{https://www.altoros.com/blog/blockchain-at-walmart-tracking-food-from-farm-to-fork/}.} are successfully using blockchain-based solutions to track products. The main drawbacks identified by the authors in using blockchain are related to scalability and blockchain interoperability.
Melesse {\em et al.} in \cite{MDR21} present a survey on digital-twin models. 
Their work identifies as one of the main problems in digital-twin models the inaccuracy in representing a physical good in a digital manner, reflecting the fact that the digital representation of a good or an asset is not always a realistic representation of the physical twin. They  discuss the use of blockchain technology as a tool to improve the effectiveness of traceability and the security of the process. Moreover, blockchains are considered by the authors a key solution to monitor a physical object from production to after sales.
Liu {\em et al.} in \cite{LYQG22}  show  how digital twins and blockchain technology together can revolutionise supply-chain management by improving the security and
efficiency of data processing, storage, and exchange. In their work Liu {\em et al.} show how blockchain technology  can be used to increase the reliability of digital twins through  security, traceability, and transparency of data. To show the actual state of the art, Liu {\em et al.} provide a comprehensive literature review for blockchain-based supply-chain management,  investigate the benefit obtained using blockchain in such processes, and give recommendations for future research directions. As future research opportunities, Liu {\em et al.} indicates the possibility of using permissioned blockchains like Hyperledger Fabric \cite{ABBCCC18} focusing also on the need of combining on-chain and off-chain data.

\ignore{
Recently Ehsan {\em et al.}\cite{EKRIAUA22} proposed a conceptual model for blockchain-based supply chain system for the agricultural sector. The model proposed by Ehsan {\em et al.} uses a permissioned blockchain to obtain a fully decentralized tracing of agricultural goods. Even if their system explains the advantages of using blockchains and smart contracts, they do not present a formalization of the functionalities that the supply chain model should satisfy and do not discuss the need for a validation procedure to check the consistency among the digital and the physical good.
}
\section{System Overview}
Our blockchain-based tracing system allows supply chain participants to record information on goods at each stage of the supply chain. Each physical good has a unique digital twin that is identified by an asset on the blockchain. The asset is associated to a current owner, a list of events and some auxiliary information. The owner of an asset
identifies the party that is supposed to own the physical good. The events represent the life-cycle of  the physical good and the auxiliary information gives additional information on the asset (e.g., location, type of good).

In our system, access-control policies are defined to prevent non-legitimate participants, including counterfeiters, from jeopardizing the ownership of assets. Moreover,  access-control policies guarantee that only legitimate manufacturers can claim the initial ownership of new assets that are introduced in the system, and that are supposed to correspond to physical goods. 
The creation of an asset is possible only if the party invoking this functionality has a legitimate role (e.g., 
it is enrolled in the system as a farmer).

In our system, if a party wants to check the authenticity of received physical goods with respect to their digital twins, the party can perform a verification procedure. We abstract\footnote{The abstraction
admits multiple real-world instantiations depending on the use case. For instance, one can refer
to an external auditor, or even skip this step if not essential.} this step
assuming the existence of a device that, after training over a given category of objects, is able to verify that a given object belongs to the specified category. For the scope of this paper, we call this device fingerprint (FP) scanner. An example of such device for the use case of food traceability can be a biological fingerprint scanner that is able to detect if a physical good belongs
or not to a specific food category (e.g., distinguishing buffalo mozzarella vs cow mozzarella); obviously this
requires an initial training to establish the biological fingerprint of a category and
how items in this category are supposed to change over time.
Therefore in our model the possibility of validating a digital twin affects
already the phase in which a category of goods is enabled in the system.
\subsection{System Set-Up and Assets}
Let us describe the initialization procedure and the list of assets managed by the system.

At first, a genesis block is created to fix some system parameters, access-control policies and system smart contracts. After the system initialization, supply chain participants (e.g., producers, distributors, certifiers) engage in an enrolling  protocol to get authorizations for their subsequent interactions. From now on, parties can record information in the system by submitting transactions.

In our system, we have five different types of assets:
\begin{itemize}
	\item Element: a good, that can be raw, such as olives, or the result of a transformation, such as olive oil. 
	\item Category: a category to which goods can belong, such as ``Oliva Ascolana del Piceno''.
	\item Batch: a batch of goods created aggregating other goods. 
	\item Production Area: a description of the location in which goods are produced.
	\item Device: an \AFS\ associated to a party.
\end{itemize}

\subsection{Entities}\label{sec:entities}
There are seven types of participants in our system.

\begin{itemize}
	\item Registration authorities.
	These are privileged parties that initialize the system and register other participants in the system with long-term credentials. The credentials links the real-world identity of the requester to his attributes and identifier. 
	We note that in a permissioned blockchain the initialization and the governance are usually managed by a group of organizations. The access control or the signing policy, (i.e., who can make what transactions in the network) and the logic of smart contracts are decided by these organizations. Members of the governance of the blockchain may overlap or not with the registration authorities. In the last case, the registration authorities decide the criteria for adding or removing parties and are responsible for authenticating, certifying, and registering network participants.
	\item \AFS\ issuers. \AFS\ issuers are the only parties capable of releasing new \AFSs\ to make sure that the system is not polluted by verification procedures that do not guarantee some desired level of quality
of service. A party that wants an \AFS\ has to send a request to an \AFS\ issuer. The registration 
authorities decide the criteria for adding or removing \AFS\ issuers. The set of authorized \AFS\ issuers must be publicly available to all parties. Using  authorized \AFS\ issuers enforces the parties to trust the functionality that
they implement.
	\item Producers. In the real world, a producer is an entity managing the production of raw materials. Accordingly, in our system, they are users authorized to introduce new items, that represent some real-world assets or to record a production area with the associated good. When a new item is created, some auxiliary information is also stored on-chain to link the physical good with its digital twin and to reflect its journey along the supply chain. Producers can introduce a new category and its benchmark, (i.e., physical properties). A party updating this information can for instance record the transfer of an item from one location to another or the item aggregation in a package. Producers can also perform all operations allowed to the role of supply chain members, which is described below.
	\item Manufacturers. Manufacturers produce goods for sale from raw materials. They are users authorized to introduce new items obtained from scratch (e.g., a manufacturer that obtains raw materials from a producer not registered to the supply chain system) or from the transformation of existing goods. They can introduce a new category and its benchmark. They can register the production area used to transform their goods. They can also perform all operations allowed to the role of supply chain members.
	\item Certifiers. Certifiers are entities responsible for certifying the authenticity of the physical good with respect to the information encoded on-chain. An audit recorded by a certifier is different from an audit recorded by other parties. Indeed, certification bodies are entities authorized by national accreditation bodies to conduct specific types of certification audits and establish certifications for qualified companies. The certifiers verify if the goods, the systems, the personnel, and the production processes of companies satisfy the certification requirements. 
	\item Consumers. Final consumers are interested in recovering the entire history of the purchased item. To keep the presentation simple, we have assumed that a consumer has only read access to supply chain information. A consumer might also be involved as an active participant in the supply chain by reporting feedbacks on the purchased items.
	\item Supply chain members. These entities are the other participants of the supply chain, (e.g., wholesaler, retailer, distributor, repackager). They are grouped into a single category since they can perform the same operations. These entities own items and give/receive items or aggregate/disaggregate them in larger/smaller packages or update the item's data.
\end{itemize}

\paragraph{Trust assumptions}\label{sec:threat}

To keep the presentation simple, we assume a single trusted registration authority. The registration authority $\Reg$ is trusted to assign correct credentials to all parties in the system and one unique identity per participant. The registration authority is trusted to verify the correctness of the attributes of a participant before enrolling the user in the system. $\Reg$ is also trusted in verifying the veracity of the role declared by parties. This authority can be obviously decentralized, as described in Section \ref{sec:improv}.

For simplicity, we assume a single \AFS\ issuer $\DevIss$, (i.e., the producer of the devices). We assume that $\DevIss$ issues only trusted devices, in the sense that these devices are well-formed and implement the functionality correctly. This trust assumption can be relaxed, as described in Section \ref{sec:improv}. 

Producers and Manufacturers are trusted to introduce new categories and/or update existing categories.

We assume that the underlying blockchain is secure. It means that an adversary cannot compromise the data stored in the blockchain, avoid valid transactions to be appended to the ledger or make parties observe a fake ledger.

\section{Security Model}
To define security we follow the simulation paradigm considering ideal and real worlds. In particular we consider the notion of universally composable (UC) secure multi-party computation
 \cite{Canetti01}. We consider hybrid models in which  protocols access ideal functionalities to perform the computation. Informally, a protocol $\pi$ is executed in the $\Gfunc$-hybrid model if $\pi$ invokes the ideal functionality $\Gfunc$ as a subroutine.

Let $\Func$ be a functionality, and consider $n$ players $\party_1,\ldots,\party_n$ executing a protocol $\pi$ that implements $\Func$.

Intuitively, $\pi$ is secure if in an execution in the real world the adversary cannot cause more harm than an adversary in the ideal world in which a trusted party computes the functions of the ideal functionality $\Func$ on behalf of the players. Moreover, we require that this guarantee holds even if $\pi$ accesses a  functionality $\Gfunc$ as subroutine.

\paragraph{The real model}
In the real world, the protocol $\pi$ is run in the presence of an adversary $\advA$ and an environment $\env$.
At the outset, $\env$ chooses the inputs $(1^\secpar,x_i)$ for each honest player $\party_i$, and gives in input to $\advA$ an auxiliary input $z$ and inputs for corrupted parties. For simplicity, we only consider static corruptions (i.e., the environment decides who is corrupt at the beginning of the protocol).

The parties start running $\pi$, with the honest players $\party_i$ behaving as prescribed in the protocol, and with corrupted parties behaving arbitrarily driven by $\advA$.

At the end of the execution, $\env$ receives the output of honest and corrupted parties and outputs a bit. We call $\REAL_{\pi,\advA,\env}(\secpar)$ the random variable corresponding to $\env$'s guess.

To make the description of a protocol modular, the parties of the protocol in the real world can access an ideal functionality $\Gfunc$ as a subroutine. In this case, we say that $\pi$ is realized in the $\Gfunc$-hybrid model. 

\paragraph{The ideal model} 
In the ideal world, there exists a trusted party that executes the ideal functionality $\Func$ on behalf of a set of dummy players $(\party_i)_{i \in [n]}$. $\env$ chooses the inputs $(1^\secpar,x_i)$ for each honest player $\party_i$, and sends to an adversary $\advS$, referred to as the simulator, inputs of corrupted parties together with an auxiliary input $z$. Honest parties send their inputs to the trusted party, while the corrupted parties send an arbitrary input as specified by $\advS$. The trusted party executes the functionality and produces outputs to give to parties. 

Finally, $\advS$ computes an arbitrary function of the inputs and gives it to $\env$. $\env$ returns a bit. We denote by $\IDEAL_{\Func,\advS,\env}(\secpar)$ the random variable corresponding to $\env$'s guess.

\begin{definition}[UC-Secure MPC]\label{def:mpc}
	Let $\pi$ be an $n$-party protocol that implements a functionality $\Func$. We say that $\pi$ securely realizes $\Func$ in the $\Gfunc$-hybrid model in the presence of malicious adversaries if for every PPT adversary $\advA$ there exists a PPT simulator $\advS$ such that for every non-uniform PPT environment $\env$ the following holds:
	\[
	\left\{\REAL^{\Gfunc}_{\pi,\advA,\env}(\secpar)\right\}_{\secpar\in\NN} \cind \left\{\IDEAL_{\Func,\advS,\env}(\secpar)\right\}_{\secpar\in\NN}.
	\]
\end{definition}

\section{The Ideal Functionalities}\label{sec:ideal_funcs}

In this section, we define the ideal functionality realized by our system for product tracing $\Fsc$ and some other functionalities $\Fff$, $\Fledg$, $\Fsig$ and $\Fsmt$. $\Fff$ models the \AFS\ that parties can use to perform verification about the item’s quality. $\Fledg$ emulates an available transaction ledger that party can use to submit transactions or to read the current state. $\Fsig$ models a secure digital signature scheme. It allows each party to sign digital messages and verify digital signatures. Finally, $\Fsmt$ models a secure channel. This is required since our protocol, in some cases, requires a private and secure exchange of information between transacting parties.

\subsection{Supply-Chain Functionality \texorpdfstring{$\Fsc$}{Lg}}\label{subsec:FSCfunc}
$\Fsc$ models the functionality realized by our product tracing system. 

Since in the real world each supply-chain member, according to the assigned role in the supply chain, can perform different operations, $\Fsc$ requires a set of producers F, manufacturers M, certifiers C, and other SC-members O to register. For example, a manufacturer should be able to create products, but a certifier not. To perform writing operations a party $\party_i\in \left\lbrace F,M,C,O\right\rbrace $ has to be registered.

We require that:
\begin{itemize}
	\item a party $\party_i$ that wants to send a training on a product $\p$, will send to $\Fsc$ a value $\cf$ that represent the unique fingerprint of $\p$. $\cf$ can be computed in many ways, a possible way is using the $\Fff$ functionality presented below;
	\item a party $\party_i$ that wants to send the outcome of an audit of a product $\p$ to $\Fsc$ , would need only to send a
	value $\auddata$ to $\Fsc$ that represents data of the audit.  This can be computed leveraging another functionality $\Fff$ or in other ways.
\end{itemize}

Associated with each product there are the supply chain events just analyzed and a state. The state can take four different values:

\begin{itemize}
	\item $\intact$: indicates that the asset is ready for use.
	\item $\packaged$: indicates that the asset has been packaged and is not ready to be used individually. A disaggregation event must occur before becoming available again (every operation performed on an item in a batch has an impact on the entire batch, this is why we require that a disaggregation is performed before performing any operation on an item in a batch).
	\item $\trans$: indicates that a handover has started between two parties. 
	\item $\destroyed$: indicates that the asset is no longer available for future processing.
\end{itemize}

For simplicity we implicitly assume that any data sent to $\Fsc$ are verified through checks that can be trivially implemented by inspecting the information stored in $\Fsc$: 
\begin{itemize}
	\item the product stored follow rules, regulations, and standards;
	\item producers and manufacturers own the production area declared;
	\item the fingerprint $\cf$ produced during a training activity on a product can be updated only by the party that  created $\cf$;
	\item the identifier of categories, production areas, and products are unique;
	\item a producer or a manufacturer that produces a new product must be the owner of the production area of the product and must be the owner of the products at creation time;
	\item the parties perform operations only over items on which they have the rights to operate;
	\item aggregation and disaggregation operations can be performed only on owned products;
	\item the items to package in an aggregation operation must be intact;
	\item the disaggregation operation can be performed only on an intact batch and the products in the batch must be packaged;
	\item the update event of an item managed by the supply chain is performed only if the specific parameters can be updated by a party that requests the update;
	\item the party that performs an handover is the owner of the good to transfer and the good is intact;
	\item the party designated for an handover is authorized to receive the asset;
	\item the goods that parties transfer in an handover are recorded goods.
\end{itemize}

$\Fsc$ is parameterized by the algorithm  $Transform$. We do not describe formally the behavior, but the informal description follows: the algorithm $Transform$ transforms a product or a set of products in a new product. This algorithm is used by manufacturers to transform assets in a new asset. The state of assets used in the transformation is modified from $\intact$ to $\destroyed$.

We call $\cf$ the asset category fingerprint that uniquely identifies a type of asset.

The data stored by $\Fsc$ are associated to a timestamp, we call this timestamp $\tEvent$ and since $\Fsc$ is a trusted functionality, we assume that $\Fsc$ computes always the correct time.

$\Fsc$ models the sequence of events that occurs along a supply chain. Each event corresponds to an operation that a party can perform. The result of the operations are stored in a list $L$, initially empty. Before that any operation is stored, $\Fsc$ checks that data are consistent in the way specified before:

\begin{mdframed}
	\begin{center}
		\bf $\Fsc$ Ideal Functionality
	\end{center}
	$\Fsc$ stores an initially empty list $L$. Parties $\party_1,\ldots,\party_n$ that interacts with $\Fsc$ may be of type $F,M,C,O$ where $F,M,C,O$ represent producers, manufacturers, certifiers, and other SC-members. Let $\advA$ be the adversary, $\advA$ may be of type $F,M,C,O$.
	
	\begin{itemize}
		
		\item Registration.
		Whenever a party writes $(\regsc, \party_i)$, if $\party_i$ is unregistered, then $\Fsc$ marks $\party_i$ as registered and outputs $(\regscok, \party_i)$ to $\party_i$ and $\advA$. Else, $\Fsc$ sends $\bot$ to $\party_i$.
		\item Register production area. Whenever a party $\party_i\in \{F,M\}$ wants to record the place of production of a new product, writes $(\farmsc,\party_i,\fieldid,\\\categid)$. $\Fsc$ stores the following tuple $(\party_i,\fieldid,\allowbreak \categid,\farmi,\tEvent,\\\intact)$ in $L$. $\Fsc$ sends $(\farmiscok,\allowbreak\party_i,\fieldid,)$ to $\party_i$ and $\advA$.

		\item Creation. Whenever a party $\party_i\in \{F,M\}$ wants to create a new product, writes $(\creatsc,\allowbreak\party_i,\itemid,\allowbreak\fieldid)$. $\Fsc$ stores the following tuple $(\party_i,\itemid,\allowbreak\fieldid,\creati,\allowbreak\tEvent,\intact)$ in $L$. $\Fsc$ sends  $(\creatscok,\party_i,\itemid,\fieldid)$ to $\party_i$ and $\advA$.
		
		\item Transformation.	Whenever a party $\party_i\in M$ wants to create a new product instance through the combination of single or multiple products writes $(\transsc,\itemid,\allowbreak\party_i,\items,\categid)$, \\where $\items=[\itemid_1,\ldots,\itemid_l]$. $\Fsc$ checks that the state of each item is intact. $\Fsc$ runs $Transform(\allowbreak\items)$, storing in $L$ a new tuple for each $\itemid_i\in\items$ where the only changed value is the state that passes from $\intact$ to $\destroyed$. Finally, $\Fsc$ stores the following tuple $(\party_i,\itemid, \items,\categid, \transi,\\\tEvent,\intact)$ in $L$. $\Fsc$ sends $(\transscok,\allowbreak\party_i,\allowbreak\itemid,\categid,\items)$ to $\party_i$ and $\advA$.
			
		\item Training. Whenever a party $\party_i\in \{F,M\}$ wants to perform a training activity,  writes $(\trainsc,\party_i,\allowbreak\categid,\allowbreak\cf)$, $\Fsc$ stores the following tuple $(\party_i,\cf,\train,\allowbreak\tEvent,\allowbreak\categid)$ in $L$ and returns $(\trainscok,\allowbreak\party_i,\allowbreak\categid,\cf)$ to $\party_i$ and $\advA$.
		\item Audit. Whenever a party $\party_i\in\{F,M,C,O\}$  writes $(\audsc,\party_i,\itemid,\allowbreak\categid,\auddata)$, \\where $\auddata$ is a tuple containing the outcome of the audit and additional information on the audit, $\Fsc$ stores the following tuple $(\audit,\party_i,\tEvent,\allowbreak\auddata,\itemid,\\ \categid)$ in $L$ and returns $(\audsc,\party_i,\itemid\allowbreak,\auddata)$ to $\party_i$ and $\advA$.	
		\item Aggregation. Whenever a party $\party_i\in\{F,M,C,O\}$ wants to physically aggregate products together writes $(\mergsc,\party_i,\batchid,\products)$, where $\products=(\p_1,\allowbreak\ldots,\p_l)$ and $\p\in\{\itemid,\batchid\}$.
		$\Fsc$ checks that the state of each $\p_i$ is intact and stores in $L$ a new tuple for $\p_i$ where $\intact$ is replaced with $\packaged$. Finally, $\Fsc$ stores the following tuple $(\party_i,\batchid,\products\aggregation,\allowbreak\tEvent,\intact)$ in $L$. $\Fsc$ sends $(\mergscok,\party_i,\allowbreak\batchid,\products)$ to $\party_i$ and $\advA$.
	
		\item Disaggregation. 	Whenever a party $\party_i\in\{F,M,C,O\}$ writes $(\splitsc,\party_i,\allowbreak\batchid)$, $\Fsc$ recovers the list of products $products$ in $\batchid$. For each $\p_i$ in $\products$, $\Fsc$ checks that the state of each $\p_i$ is packaged and stores a new tuple where the only changed value is the state that passes from $\packaged$ to $\intact$. $\Fsc$ stores the tuple $(\party_i,\batchid,\allowbreak\products,\allowbreak\disaggregation,\tEvent,\allowbreak\destroyed)$. $\Fsc$ sends $\allowbreak(\splitscok,\allowbreak\party_i,\batchid)$ to $\party_i$ and $\advA$.
	
		\item Handover. The handover operation is divided in the following actions:
			\begin{itemize}
				\item Whenever a party $\party_i\in\{F,M,C,O\}$  writes $(\handsc,\p,\party_i,\party_j)$, where $\p\in\{\itemid,\allowbreak\batchid\}$, $\Fsc$ stores a new tuple associated with $\p$, where the state is changed from $\intact$ to $\trans$ and adds $\party_j$ in the tuple as new designated owner for $\p$. $\Fsc$ sends $(\handsc,\p,\party_i,\party_j)$ to $\party_i$, $\party_j$ and $\advA$ to indicate that a handover is in progress between $\party_i$ and $\party_j$.
				\item Whenever a party $\party_j\in\{F,M,C,O\}$ writes $(\recvsc,\p,\party_j, \party_i)$, where $\p\in\{\itemid,\allowbreak\batchid\}$, $\Fsc$ stores the new following tuple associated to $\p$ $(\party_j,\p,\allowbreak\handscok,\tEvent,\allowbreak\intact)$, where the owner is changed from $\party_i$ to $\party_j$, and the state is changed from $\trans$ to $\intact$. $\Fsc$ sends $(\recvscok,\p,\party_i,\allowbreak\party_j)$ to $\party_i$, $\party_j$ and $\advA$ to indicate that a handover occurred between $\party_i$ and $\party_j$.
				\item Whenever a party $\party_j\in\{F,M,C,O\}$ writes $(\rejsc,\p,\party_j,\party_i)$, where $\p\in\{\itemid,\allowbreak\batchid\}$, $\Fsc$ stores the following tuple associated to $\p$ $(\party_i,\p,\handscko,\allowbreak\tEvent,\intact)$. Finally, $\Fsc$ sends\\ $(\rejscok,\p,\party_i,\allowbreak\party_j)$ to $\party_i$, $\party_j$ and $\advA$ to indicate that a handover has failed between $\party_i$ and $\party_j$.
				
			\end{itemize}
		
		\item Update.
		Whenever a party $\party_i\in\{F,M,C,O\}$ writes $(\updsc,\party_i,\p,\\ \newstate)$, where $\p\in\{\itemid,\batchid,\fieldid\}$, and $\newstate$ is the update to the state of $\p$, $\Fsc$ stores a new tuple associated to $\p$ where the only change is the state that is updated with value $\newstate$. $\Fsc$ sends $(\updscok,\party_i,\p,\newstate)$ to $\party_i$ and $\advA$.
		
		\item Read.
		Whenever a party $\party_i\in\{F,M,C,O\}$ writes $(\readsc,\assetid)$, where $\assetid\in\allowbreak\{\itemid,\allowbreak\batchid,\allowbreak\fieldid,\allowbreak\categid\}$ from a party $\party_i$ or the adversary $\advA$, $\Fsc$ returns a list of all records stored associated to $\assetid$.
	\end{itemize}
\end{mdframed}

The Update process in $\Fsc$ is required since this process describes the update of a product instance. This operation can be performed only by the current owner of the product and only for specific fields of information. A particular case is the deletion of a product instance from the supply chain. A deletion event is essentially the inverse of a creation event: an asset is invalidated and can no longer be used. It can occur when products are sold to the end-users or there is a product recall or loss due to contamination or accident. In case of a product recall, the product will be picked up from the market by competent authorities. Finally, the product can be destroyed by competent authorities with an update event, which might provide also proof of destruction. By inspecting the item's history, it will be possible to monitor the entire recall process.

\subsection{AssetFingerPrint Scanner Functionality \texorpdfstring{$\Fff$}{Lg}}\label{subsec:FFFfunc}
In the following, we propose the functionality $\Fff$ that will model the physical \AFS.

We make the following trust assumptions on the scanner.
\begin{itemize}
	\item The parties can verify the authenticity of the hardware (e.g., the integrity and no-tampering of the hardware). This means that an adversary is unable to replace the \AFS\ with a malicious one.
	\item The \AFS\ can be uniquely identified.
	\item The output is linked to the \AFS\ used.
\end{itemize}

$\Fff$ handles the following operations.
\begin{itemize}
	\item Initialize a device.
	\item Perform the training operations to allow $\Fff$ to recognize a specific category of products.
	\item Verify that a given item belongs to the claimed category.
	\item Transfer a devise from party $\party_i$ to party $\party_j$.
	\item List all \AFSs\ issued to parties. 
\end{itemize}

Moreover, we require that the training and update of a training can be performed only by the party that creates the asset fingerprint $\cf$.

We assume that an adversary has access only to the \AFSs\ of the corrupted parties. Each new scanner has associated a tuple $(\party_i,\devid)$, indicating that the scanner belongs to $\party_i$ and is identified by $\devid$. The device can be transferred from a party to another. In this case we assume that once the handover happens, the owner of the device is updated. We assume that every time that operations of $\Fff$ are called, the caller is the owner of the device on which the operation is performed and in the case of an handover that the device is not jet engaged in another handover procedure.

The $\Fff$ functionality uses $Training$ and $Verify$ algorithms. $Training$ can be seen as a procedure able to extract a biochemical and molecular mapping of a category, a kind of fingerprint, based on sample products. We call the output of $Training$ algorithm $\cf$ and it is the specific fingerprint of the category on which the training was performed. This process is time-consuming and might make unusable the products on which it is performed. $Verify$ takes an input a physical product and, thanks to some evaluations, can verify if the product belongs to a category. $Verify$ has potentially a shorter process time than $Training$ and does not necessarily make unusable the products on which it is performed. Since the verification process requires that the function is executed on a physical good, we require that also in the ideal world, the query is performed sending in input the physical good, meaning that it is not possible to execute the $Verify$ algorithm, if there is no good to check.

\begin{mdframed}
	\begin{center}
		\bf $\Fff$ Ideal Functionality
	\end{center}
	$\Fff$ runs with parties $\party_1,\ldots,\party_n$ and adversary $\advA$. It begins with a list $L$ initialized to $\emptyset$.
	\begin{itemize}
		\item Init \AFS.
		Whenever a party $\party_i\in\{\DevIss\}$ writes $(\initff,\party_i)$, $\Fff$ draws $\devid\allowbreak\gets Sample(1^\secpar)$. $\Fff$ stores the following tuple in $L$: $(\party_i,\devid,)$ and sends $(\allowbreak\initffok,\devid)$ to $\party_i$ and $\advA$.

		\item Training. Whenever a party $\party_i$ writes $(\trainff,\allowbreak\categid)$, then:
			\begin{itemize}
				\item If the training is successful and:
				\begin{itemize}
					\item A tuple $(\categid,*)$ does not exist, $\Fff$ stores the following tuple $(\categid,\allowbreak\cf)$ in $L$, where $\cf\gets Training(1^\secpar,\categid)$ and sends $(\initffok,\categid,\allowbreak\cf)$ to party $\party_i$ and $\advA$;
					\item A tuple $(\categid,*)$ already exists, $\Fff$ updates the following tuple $(\categid,*)$ in $L$ to $(\categid,\cf)$, where $\cf\gets Training(1^\secpar,\allowbreak\categid)$ and sends $(\updffok,\categid,\cf)$ to party $\party_i$ and $\advA$;
				\end{itemize}
				\item If the training is not successful, $\Fff$ sends $(\trainffko,\allowbreak\categid)$ to party $\party_i$ and $\advA$.
			\end{itemize}
		\item Verify. Whenever a party $\party_i$ writes $(\evalff,\party_i,\itm,\allowbreak\categid)$, where $\itm$ is a physical product, then $\Fff$ runs $\res\gets Verify(1^\secpar,\itm,\allowbreak\categid)$ and:
		\begin{itemize}
			\item If the verification has produced a result, $\Fff$ sends $\allowbreak(\evalffok,\allowbreak\res,\allowbreak\devid,\allowbreak\itm,\allowbreak\categid)$ to party $\party_i$ and $\advA$.
			\item Otherwise, $\Fff$ sends $(\evalffko,\itm,\allowbreak\categid)$ to party $\party_i$ and $\advA$.
		\end{itemize}
		\item Handover. Whenever a party $\party_i$ writes $(\handff,\allowbreak\devid,\party_i,\party_j,)$, $\Fff$ stores the tuple $(\recvffok,\allowbreak\devid,\party_i,\party_j)$ in $L$ and adds the tuple $(\party_j,\allowbreak\devid)$ in $L$. $\Fff$ sends $(\recvffok,\devid,\allowbreak\party_i,\allowbreak\party_j)$ to $\party_i$, $\party_j$ and $\advA$ to indicate that a handover has occurred between $\party_i$ and $\party_j$.
		\item Device withdrawal.
		Whenever the party $\party_i\in\{\DevIss\}$ that creates a device $\devid$ writes  $(\deldev,\allowbreak\devid)$, $\Fff$ stores the tuple $(\deldevok,\allowbreak\party_i,\devid)$ in $L$ and sends $(\deldevok,\party_i,\devid)$ to $\party_i$ and $\advA$.
		\item Read.
		Upon input $\readff$ from a party $\party_i$ or $\advA$, $\Fff$ returns to the caller a list of all records stored in $L$, which matches parties to devices (e.g., of type $(\party_i,\devid,*)$).
	\end{itemize}
\end{mdframed}

\subsection{Ledger Functionality \texorpdfstring{$\Fledg$}{Lg}}\label{subsec:FLfunc}
$\Fledg$ provides the abstraction of a transaction ledger. 

We assume that every submitted transaction is validated before they are added to the list of transactions. The validation algorithm checks that the transaction are not wrongly formatted and that the transaction was not submitted in the past. The validation algorithm is required to guarantee the chain quality.

Anyone might request a read of the list $L$.  We assume that the ledger is timestamped, live and immutable: a valid transaction submitted to the ledger will eventually be included and cannot be deleted afterwards. In a nutshell, valid transactions are appended, final, and available to all parties \cite{ACCDET20}. In blockchains with forks, a transaction can be considered final when it has achieved sufficient confirmations or has reached a checkpoint. To achieve strong guarantees, a party might issue transactions relative to an account: a transaction might contain an identifier $\party$, which is the abstract identity that claims ownership of the transaction. We can represent this situation by having transactions $\tx$ be pairs $(\tx',\party)$ with the above meaning\cite{BMTZ17}. $\Fledg$ has a validation predicate. 

\begin{mdframed}
	\begin{center}
		\bf $\Fledg$ Ideal Functionality
	\end{center}
	$\Fledg$ runs with parties $\party_1,\ldots,\party_n$ and adversary $\advA$. $\Fledg$ stores an initially empty list $L$ of bit strings.
	\begin{itemize}
		\item Submitting a transaction. Upon receiving $I=(\submitll,\allowbreak\tx)$ from a party $\party$, $\Fledg$ appends $\tx$ to $L$ and returns $(\submit,\tx,\party)$ to $\party$. If $\party$ is corrupt, then $\Fledg$ returns $(\submit,\tx,\party)$ to $\advA$.
		\item Reading the state. Upon receiving $I=\readll$ from a party $\party$, $\Fledg$ returns $(\rd,L)$ to $\party$. If $\party$ is corrupt, then $\Fledg$ returns $(\rd,L)$ to $\advA$.
	\end{itemize}
\end{mdframed}

\subsection{Digital Signature Functionality \texorpdfstring{$\Fsig$}{Lg}}\label{subsec:FSIGfunc}

$\Fsig$ models a digital signature scheme and consequently is parameterized by three algorithms $KeyGen$, $Sign$ and $Verify$. We describe $\Fsig$ using the same syntax of \cite{ACCDET20}.

The signature scheme is unforgeability under
chosen-message attack\cite{GMR88}, meaning that given an oracle that generates valid signatures, it is infeasible for an efficient adversary to output a valid signature on a message that has not been queried to the oracle.

\begin{mdframed}
	\begin{center}
		\bf $\Fsig$ Ideal Functionality
	\end{center}
	$\Fsig$ requires $sid=(Signer,sid')$, where $Signer$ is the party identifier of the signer. Set $C$, initially empty, specifies the set of currently corrupted parties. The functionality keeps a set $\set$ of properly signed messages.
	\begin{itemize}
		\item Key generation. Upon the first activation from Signer, $\Fsig$ runs $(\pk,\sk)\gets KeyGen(\secpar)$, where $\secpar$ is obtained from the security parameter tape, and stores $(\pk,\sk)$.
		\item Send public key. Upon input $pubkey$ from a party $\party$, $\Fsig$ outputs $(pubkey,\allowbreak\pk)$ to $\party$.
		\item Sign. Upon input $(sign,m)$ from party $Signer$ with $m\in\bin^*$, $\Fsig$ computes $\sgn\gets sign(\sk,m)$. Finally, $\Fsig$ sets $\set\gets\set\cup\{m\}$ and outputs $\sgn$ to $Signer$.
		\item Verify.  Upon input $(verify,\pk',m',\sgn')$ from a party $\party$, $\Fsig$ computes $b\gets verify(\pk',\allowbreak m',\allowbreak\sgn')$. If $Signer\not\in C\land\pk=\pk'\land b\land m'\not\in\set$ then $\Fsig$ outputs $(\res,0)$ to $\party$. Else, $\Fsig$ outputs $(\res,b)$ to $\party$.
		\item Corruption.  Upon input $(corrupt,\party)$ from the adversary, $\Fsig$ sets $C\gets C \cup\{\party\}$. If $\party= Signer$, then $\Fsig$ additionally outputs $\sk$ to $\advA$.
	\end{itemize}
\end{mdframed}

\subsection{Secure Message Transmission Functionality \texorpdfstring{$\Fsmt$}{Lg}}\label{subsec:FSMTfunc}
$\Fsmt$ models a secure channel between a sender $S$ and a receiver $R$.  We describe $\Fsig$ using the same syntax of \cite{ACCDET20}.

\begin{mdframed}
	\begin{center}
		\bf $\Fsmt$ Ideal Functionality
	\end{center}
	$\Fsmt$ is for transmitting messages in a secure and private manner.
	\begin{itemize}
		\item Send message. Upon input $(\sendsmt,\party_j,m)$ from a party $\party_i$:
		\begin{itemize}
			\item If both $\party_i$ and $\party_j$ are honest, $\Fsmt$ provides a private delayed output $(\sendsmtok,\allowbreak\party_i,\allowbreak\party_j,m)$ to $\party_j$.
			\item If at least one of $\party_i$ and $\party_j$ is corrupt, $\Fsmt$ provides a public delayed output $(\sendsmtok,\allowbreak\party_i,\allowbreak\party_j,m)$ to $\advA$.
		\end{itemize} 
	\end{itemize}
\end{mdframed}

\section{The Real-World: Our System \texorpdfstring{$\Psc$}{Lg}}
This section describes how our system works in the real world. Our supply chain protocol $\Psc$ realizes the functionality $\Fsc$ and operates with functionalities $\Fledg$, $\Fsmt$, $\Fsig$ and $\Fff$.

    As explained in Section \ref{subsec:FLfunc}, $\Fledg$ models a transaction ledger. To achieve strong guarantees, a party issues transactions relative to an account. Signatures enter the picture at this level: an honest participant of the network will issue only signed transactions on the network. Each party $\party$ has a signature key pair $(\pk,\sk)$ obtained via $\Fsig$.

Whenever $\party$ is supposed to submit a transaction $\tx'$, $\party$ signs it and appends the signature $\sgn$ and its party identifier $\party$. We define the transaction format as $\tx=(\tx',\sgn,\party)$, where the first part is an arbitrary transaction. To verify the signature, it is important to provide the verification key. This information must be added in the transaction or must be possible to calculate it from the transaction data. We assume that the public key for the signature verification can be extracted from the party identifier $\party$. In the real world, any blockchain system that realizes $\Fledg$ can be used. For our purpose, a reasonable choice is Hyperledger Fabric\cite{ABBCCC18}, which offers a membership service infrastructure \footnote{See the Membership Service Providers (MSP) at \url{https://hyperledger-fabric.readthedocs.io/en/release-2.2/msp.html} (accessed 14 July 2022).} to grant long-term credentials to network participants and supports a pluggable transaction validation \footnote{See the Pluggable Transaction Endorsement and Validation at \url{https://hyperledger-fabric.readthedocs.io/en/release-2.2/pluggable_endorsement_and_validation.html} (accessed 14 July 2022).}. The functionality $\Fsmt$ can be implemented by establishing a secure communication channel between parties (e.g., via TLS protocol). $\Fsig$ can be accommodated by employing a digital scheme that secure realizes the functionality (e.g., ECDSA). In the real world, the functionality $\Fff$ is realized by physical \AFSs\ that are able to verify the originality of a product. We assume in $\Fff$ that the list of all allocated devices is available to all parties. This can be achieved in the real world by involving an \AFS\ issuer $\DevIss$. If a party wants to be provided with at least one \AFS\ to be able to perform quality checks on the items, the party has to request $\DevIss$. When $\DevIss$ issues a new \AFS, it makes publicly available which \AFS\ is associated with which party.

\subsection{The Device Protocol \texorpdfstring{$\Paf$}{Lg}}

The $\Fff$ functionality can be implemented for instance leveraging
a scanner provided by a device producer $\DevIss$. 

\subsubsection{Trust Assumptions on \AFS.}\label{subsec:trust_ass}
In the following we specify which assumption can be done on the \AFS.

Obviously an adversary can not be forced to use well formed
\AFSs\ and thus we have made the following trust assumptions on the \AFS\ in Section \ref{subsec:FFFfunc}:
\begin{itemize}
	\item Honest parties can verify the authenticity of the hardware (e.g., the integrity of its behavior). This means that an adversary is unable to undetectably replace the \AFS\ of a honest party with a malicious one.
	\item The output of a genuine \AFS\ can be linked to it.
	\item The adversary can not produce fake outputs on behalf of a genuine
\AFS.
\end{itemize}

We note that these assumptions might be reached in the real world with tamper-proof hardware. Tamper-proof hardware has many similarities with our case study: an honest party can create a hardware token implementing any desired polytime functionality and an adversary given the token can only observe the input/output of this token~\cite{Katz07}. In our case study, the hardware token must be completely tamper-proof if we want to guarantee the first assumption. Typically
(e.g., smart cards), tamper-proof hardware can also generate digital signatures
to certify the outputs that it produces.
 This makes a \AFS\ uniquely identifiable and capable of authenticating
its outputs.

A \AFS\ is required to support the following operations.
\begin{itemize}
	\item Training operation. Given a good, the device is able to extract the fingerprint of the good that is digitally stored. We call $\cf$ the digital description of the fingerprint.
	\item Verification. Given a physical good, and a fingerprint  description $\cf$, the device is able to check that the given good belongs to the category of goods with fingerprint $\cf$.
\end{itemize}

\subsubsection{\texorpdfstring{$\Paf$}{Lg} Description.}
We report the $\Paf$ protocol. The operation performed by the protocol are stored on a ledger. The protocol $\Paf$ is the following:
\begin{itemize}
	\item Once the device is created, $\DevIss$ submits the transaction $\tx=((\creat,\allowbreak\devid,\DevIss),\sgn,\allowbreak(\DevIss,\pk_\DevIss,\attr_\DevIss,\sgn_\DevIss))$ via $\Fledg$, to indicate that the device is created. To protect the validity of transactions, $\DevIss$ generates the signatures of the content of the transaction through invocation to $\Fsig$ and append the received signature together with user credentials to the transaction.
	\item Once the device is created, $\DevIss$ handover the device to the party $\party_i$ that need it. To record this operation $\DevIss$ submits the transaction $\tx=(	(\hand,\allowbreak\devid,\DevIss,\party_i),\sgn,(\DevIss,\pk_\DevIss,\allowbreak\attr_\DevIss,\sgn_\DevIss))$ which indicates that an \AFS\ has been sent to $\party_i$. $\DevIss$ generates $\sgn$ through invocation to $\Fsig$.
	\item Any party $\party_i$ that owns the device can decide to handover the device to a party $\party_j$. To record this operation $\party_i$ submit the transaction $\tx=(	(\hand,\allowbreak\devid,\party_i,\party_j),\sgn,\allowbreak(\party_i,\pk_{\party_i},\attr_{\party_i},\sgn_{\party_i}))$ which indicates that an \AFS\ has been sent to $\party_j$. $\party_i$ generates $\sgn$ through invocation to $\Fsig$.
	\item Any party $\party_i$ that owns the device can perform a training on a good to obtain the fingerprint $\cf$ of the category of the good. The device performs this operation using information of the good to analyze. To record this operation, $\party_i$ submits transaction $\tx=(	(\categid,\cf),\sgn,(\party_i,\pk_{\party_i},\allowbreak\attr_{\party_i},\sgn_{\party_i}))$. $\party_i$ generates $\sgn$ through invocation to $\Fsig$.
	\item Any party $\party_i$ that owns the device can perform a verification on a physical asset given the expected category identifier $\categid$ at which the asset belongs. The device performs this operation using information of the good to analyze. $\party_i$ reads the ledger to obtain the fingerprint $\cf$ associated to $\categid$ and performs the analysis with $\cf$. This operation returns a digital information $\auddata$.
	\item If the device issuer $\DevIss$ that issues a device $\devid$ needs to withdraw this device from the marked (i.e., this device is broken), $\DevIss$ submits the following transaction $\tx=((\deldev,\allowbreak\devid,\DevIss),\sgn,(\DevIss,\allowbreak\pk_\DevIss,\attr_\DevIss,\sgn_\DevIss))$ via $\Fledg$. This transaction indicates that the device cannot be used anymore by parties in the system.
	\item Every user can read the blockchain to obtain information about devices and asset categories stored.
\end{itemize}

We have the following assumption: 
\begin{itemize}
	\item the creation of devices is always performed honestly and no device issuer misbehaves;
	\item the device withdrawal of a device is always performed honestly by the device issuer that issues the device;
	\item the training process is always performed honestly since the process is supervised by an honest controller.
\end{itemize}

The only actions that an adversary $\advA$ can try to perform in the real word is to impersonate a device issuer to issue misbehaving devices or set to $\deldevok$ devices that have not really been withdrawn. This attack cannot be performed except that $\advA$ is able to break the unforgeability of the signature scheme.

\subsection{The Protocol \texorpdfstring{$\Psc$}{Lg}}
In the following we describe the protocol $\Psc$.

The building block of the protocol $\Psc$ is an SC-smart contract $\smartct$, that manages the supply chain activities.

Each asset in the system has a unique identifier. The identifier is used to capture important information on the movement of assets along the supply chain and might be of the form E2344..., where the first character identifies the asset's type, in this case, Element, and the remaining characters uniquely identify the asset among this type. The remaining characters may be obtained, for example, by randomly generating a value $r\gets\bin^\secpar$ that has to be unique. The identifier may also be obtained by using existing open standards such as the Electronic Product Code (EPC).

Access control policies must be defined to model the real-world constraint where each party, according to its role in the supply chain, can perform different operations on an asset. As explained in Section \ref{sec:entities}, we assume the existence of a single trusty registration authority $\Reg$ responsible of register other participants in the system with long-term credentials. $\Reg$ has a unique unforgeable identity that is available to all nodes in the network. At setup, $\Reg$ generates a key pair corresponding to the $\Reg$ identifier. The public information of $\Reg$ is stored in the genesis block. The smart contract and the access control policies realize our functionality $\Fsc$. They are initialized during the setup phase of the system by an authorized group member (i.e., registration authorities). In a nutshell, the protocol's initiation consists of the setup of the blockchain global parameters. The $\smartct$ handles all the transaction describes below and is executed by blockchain nodes. We assume that the smart contract correctly updates the state accordingly to the transaction retrieved by the ledger (i.e., is trusted to execute correctly). Before applying the state changes that come with the transaction itself, the smart contract performs various checks to be sure that the transaction satisfies the access control policies and the logic that the contract implements. In short, $\smartct$ verifies the authorization of the transaction submitter and stores authenticated and verified tracing information on a product. In our case study, for example, we have the following constraints.
\begin{itemize}
	\item Specific transactions can be performed only by parties with a particular role (i.e., access control policies).
	\item Any write operation on an item can be carried out only by the current owner.
	\item The product data must satisfy the specification of the validity of the product. So, $\smartct$ might perform automated regulatory-compliance verifications by checking that all the information required by regulation is provided.
	\item Specific operations can be performed only on specific types of asset (e.g., an aggregation can occur only for Batch or Element but not for Product Area or Category).
\end{itemize}

The protocol starts with parties' registration since registration is needed for all writing operations. When a participant sends a registration request with attributes certification to $\Reg$, the registration authority verifies the received information and registers the participant in the system with long-term credentials. Let us start by describing the protocol machine for $\Reg$. $\Reg$ has a bit $\initialized\gets false$ and an initially empty list $L$ of certified parties' credentials. We assume private channels between the registration authority and the parties.

\begin{itemize}
	\item Upon input $\init$, $\Reg$ checks if $\initialized$ is set. If set, $\Reg$ aborts. Else, $\Reg$ generates a signature key pair $(\pk_R,\sk_R )$ via invocations to $\Fsig$ and $\Reg$ stores, $\pk_R$, $\sk_R$. The public information of the registration authority, the identifier and the public key $\pk_R$ are stored in the genesis block of the ledger. $\Reg$ sets $\initialized\gets true$ and set the list of registered users $L$ to empty.
	\item Upon receiving $(\reg,\party_i,\pk_i,\attr_i)$ from $\party_i$ via $\Fsmt$, where $\attr_i$ are party's attributes,  $\Reg$ signs $(\party_i,\pk_i,\attr_i)$ and obtains the signature $\sgn_i$ that will be used by $\party_i$ together with $\pk_i,\attr_i$ as his credential for the system. The signature $\sgn_i$ is computed through invocation to $\Fsig$. $\Reg$ stores $(\party_i,\pk_i,\attr_i,\sgn_i)$ in $L$ and sends the message $(\regok,\allowbreak\party_i,\allowbreak\pk_i,\allowbreak\attr_i,\allowbreak\sgn_i)$ via $\Fsmt$ to $\party_i$.
\end{itemize}

The protocol for $\DevIss$ is the following. At setup, $\DevIss$ generates a key pair corresponding to its identifier and requests $\Reg$ for registration and attributes certification through a secure channel. If the authentication is successful, $\DevIss$ receives its party credentials $\sgn_\DevIss$ with which $\DevIss$ can issue new \AFSs\ in the system. We describe the protocol for a party $\DevIss$ of the system. $\DevIss$ has a bit $\regok\gets false$.

\begin{itemize}
	\item Upon input $\reg$, $\DevIss$ generates a signature key pair $(\pk_\DevIss,\sk_\DevIss)$ via invocations to $\Fsig$ and a unique identifier $\DevIss$. It sends a message $(\reg,\DevIss,\allowbreak\pk_\DevIss,\attr_\DevIss)$ to $\Reg$ via $\Fsmt$. If all steps succeeded, then $\DevIss$ sets $\regok\gets true$ and stores  $\DevIss$, $\pk_\DevIss$, $\sk_\DevIss$, $\attr_\DevIss$, $\sgn_\DevIss$. The registration phase is mandatory for each party that needs to submit transactions.
	\item $\DevIss$ creates new devices calling $\Fff$ with input $(\initff,\DevIss)$. The operation will produce a device identifier $\devid$. 
	\item If $\DevIss$ needs to withdraw a device produced by himself from the market (e.g., this specific device type is not working properly), $\DevIss$ calls $\Fff$ with input $(\deldev,\devid)$. This transaction indicates that an \AFS\ cannot be used anymore by parties in the system.
\end{itemize}

Let us now describe the protocol for a party $\party_i$ not already analyzed before. At the setup, $\party_i$ generates a key pair corresponding to its identifier and requests $\Reg$ for registration and attributes certification through a secure channel. If the authentication is successful, $\party_i$ receives its party credentials $\sgn_i$ with which $\party_i$ can perform operations in the system if it satisfies the access control policies. $\party_i$ has a bit $\regok\gets false$.

$\Psc$ for a party $\party_i$ is the following.

\begin{itemize}
	\item Upon input $\reg$, $\party_i$ generates a signature key pair $(\pk_i,\sk_i)$ via invocations to $\Fsig$ and a unique identifier $\party_i$. It sends a message $(\reg,\party_i,\allowbreak\pk_i,\attr_i)$ to $\Reg$ via $\Fsmt$. The value $\attr_i$ contains the role of party together with other additional information. If all steps succeeded, then $\party_i$ sets $\regok\gets true$ and stores  $\party_i,\pk_i,\sk_i,\attr_i,\sgn_i$.
	\item When $\party_i$ wants to be provided with an \AFS, $\party_i$ sends the message $(\dev,\party_i)$ to $\DevIss$ via $\Fsmt$ and $\DevIss$ sends message $(\handff,\DevIss,\party_i)$ to $\Fff$.
	\item If $\party_i$ wants to create a new element, $\party_i$ creates the transaction $\tx'=\allowbreak(\creat,\allowbreak\itemid,\allowbreak\categid,\allowbreak\fieldid)$. After that, $\party_i$ can submit the transaction $\tx=\allowbreak(\tx',\allowbreak\sgn,\allowbreak(\party_i,\allowbreak\pk_i,\allowbreak\attr_i,\allowbreak\sgn_i))$ via $\Fledg$. The signature $\sgn$ of $\tx'$ is obtained through invocation to $\Fsig$.
	\item If $\party_i$ wants to record a new product area, it submits the transaction $\tx=\allowbreak((\farm,\allowbreak\fieldid),\allowbreak\sgn,\allowbreak(\party_i,\allowbreak\pk_i,\allowbreak\attr_i,\allowbreak\sgn_i))$. The signature $\sgn$ of $\tx'$ is obtained through invocation to $\Fsig$. 
	\item If $\party_i$ wants to create a new element, that is obtained by combining other stored items, it submits the transaction $\tx=(\allowbreak(\transform,\allowbreak\itemid,\allowbreak\items,\allowbreak\categid),\allowbreak\sgn,\allowbreak(\party_i,\allowbreak\pk_i,\allowbreak\attr_i,\allowbreak\sgn_i))$, where $\items$ is an array of item identifiers $\items=\left[(\itemid_1),\ldots,(\itemid_l)\right]$. The signature $\sgn$ of $\tx'$ is obtained through invocation to $\Fsig$.
	\item If $\party_i$ performs a training on a category $\categid$, $\party_i$ sends $(\trainff,\categid)$ to $\Fff$. If the training has returned a result $\cf$, $\party_i$ submits the transaction $\tx=(\allowbreak(\training,\allowbreak\categid,\allowbreak\cf),\allowbreak\sgn,\allowbreak(\party_i,\allowbreak\pk_i,\allowbreak\attr_i,\allowbreak\sgn_i))$. The signature $\sgn$ of $\tx'$ is obtained through invocation to $\Fsig$.
	\item Every time $\party_i$ inspects an item using an allocated \AFS\footnote{The \AFS\ has to be valid and not withdrawn by the market.}, $\party_i$ performs the following operation. First executes the verification sending $(\evalff,\party_i,\itemid,\allowbreak\categid)$ to $\Fff$. If the verification has returned a result $(\res,\devid,\itemid,\categid)$, submits the transaction $\tx=((\audit,\itemid,\auddata),\sgn,\allowbreak(\party_i,\pk_i,\attr_i,\sgn_i))$, where $\auddata=\allowbreak(\devid,\allowbreak\res,\itemid,\allowbreak\categid)$. Before accepting the transaction $\tx$, $\smartct$ query $\Fff$ to check that the device with identifier $\devid$ is valid and is not withdrawn. If during the verification phase has been chosen to return a more complex output, such as the value of the parameters evaluated, then it is possible to save on-chain this information for a more complete audit. If $\party_i$ is a Certifier, the audit recorded will be more accurate.
	\item If $\party_i$ wants to aggregate different stored items in a single package, it submits the transaction $\tx=((\merge,\batchid,\products),\sgn,(\party_i,\pk_i,\attr_i,\sgn_i))$, \\where products is an array containing the item identifier. The signature $\sgn$ of $\tx'$ is obtained through invocation to $\Fsig$.
	\item If $\party_i$ wants to disaggregate a package to retrieve stored items, it submits the  transaction $\tx=(\allowbreak(\scsplit,\allowbreak\batchid),\allowbreak\sgn,\allowbreak(\party_i,\pk_i,\attr_i,\sgn_i))$. The signature $\sgn$ of $\tx'$ is obtained through invocation to $\Fsig$.
	\item If $\party_i$ wants to update an asset's state, it submits the transaction $\tx=((\upd,\p,\allowbreak\newstate),\sgn,(\party_i,\pk_i,\allowbreak\attr_i,\sgn_i))$, where  $\newstate$ is the update of the state of the product that can be set from $\intact$ or $\packaged$ to $\destroyed$. The signature $\sgn$ of $\tx'$ is obtained through invocation to $\Fsig$.
	\item If $\party_i$ wants to transfer an asset to another party $\party_j$, it should do this by submitting the transaction $\tx=((\hand,\p,\party_i,\party_j),\sgn,(\party_i,\pk_i,\attr_i,\sgn_i ))$. If $\party_j$ wants to reject an asset from $\party_i$ that was sent to $\party_j$, $\party_j$ can do this by submitting the transaction $\tx=((\rejok,\p,\party_j,\party_i),\sgn,(\party_j,\pk_j,\attr_j,\sgn_j))$. If $\party_j$ wants to accept an asset from $\party_i$ that was sent to it, $\party_j$ can do this by submitting the transaction $\tx=((\recvok,\p,\party_j,\party_i),\sgn,(\party_j,\pk_j,\attr_j,\sgn_j))$. \ignore{If $\party_j$ does not receive the asset, $\party_j$ can reject it and the ownership of the good returns to $\party_i$.}
	\item The party $\party_i$ transmits only transactions. When $\party_i$ wants to retrieve an asset state, it relies upon blockchain consensus nodes, who process transactions. In this case, $\party_i$ sends the following request $(\readsc,\assetid)$ to one or multiple consensus nodes via $\Fsmt$ to obtain a list of all value changes of $\assetid$. As explained before, $\party_i$ can perform read operations even if is not registered in the system.
\end{itemize}

All data that a party $\party$ wants to transmit to $\Fledg$ are previously signed through invocations to $\Fsig$. The data, together with the signature and $\party$ credentials, are submitted to $\Fledg$. When creating a transaction, the data entry can be automated and speed up through IoT devices, QR code scanning, and so on. 

\subsection{Security of $\Psc$}\label{sec:security}


\begin{theorem}\label{thm:sec}
For any PPT adversaries $\advA$, there exists a PPT simulator $\advS$ such that for all non-uniform PPT distinguisher $\advD$, the protocol $\Psc$ securely computes $\Fsc$ in the $(\Fff,\Fledg,\Fsmt,\Fsig)$-hybrid model.
\end{theorem}

\begin{proof}
We need to show that for every PPT adversary $\advA$, there exists a PPT simulator $\advS$ such that no non-uniform PPT distinguisher $\advD$ can tell apart the experiments \\$\REAL^{\Fff,\Fledg,\Fsmt,\Fsig}_{\Psc,\advA,\advD}(\lambda)$ and $\IDEAL_{\Fsc,\advS,\advD}(\lambda)$.

We define a simulator $\advS$ that interact with $\Fsc$. We consider only static corruptions, (i.e., the choice of which players to attack is made before the protocol starts). $\advS$ emulates functionalities $\Fledg$, $\Fsig$, $\Fsmt$ and $\Fff$. $\advS$ initially sets $\initialized\gets false$.

\begin{itemize}
	\item Simulating parties' registration. Upon receiving\\ $(\regscok,\party_i)$ from $\Fsc$, $\advS$ verifies that $\initialized$ is set to $true$. If the verification fails, $\advS$ first simulates the system initialization. $\advS$ creates and stores, as the real $\Reg$, a signature key pair $\pk_R^*,\sk_R^*$ generated using $\Fsig$. Finally, $\advS$ sets $\initialized\gets true$. Once that $\initialized\gets true$, if $\party_i$ is honest, $\advS$ generates a signature key pair $(\pk_i^*,\sk_i^*)$ via the simulated $\Fsig$ as the real $\party_i$ and marks $\party_i$ as registered. If $\party_i$ is corrupt, $\advS$ initializes $(\pk_i^*,\sk_i^*)$ with the received values. $\advS$ runs the $\Reg$ algorithm with input $(\reg,\party_i,\pk_i^*,\attr_i)$.
	In both cases, $\advS$, as the real $\Reg$, signs the message $(\party_i,\pk_i^*,\attr_i)$ via the simulated $\Fsig$ to obtain $\sgn_i^*$ and sends $\sgn_i^*$ to $\party_i$. $\advS$ marks $\party_i$ as registered and internally stores $(\party_i,(\pk_i^*,\sk_i^*),\sgn_i^*)$.
	
	\item Simulating honest regular parties. The actions of honest parties are simulated as follows.
	\begin{itemize}
		\item Upon input
		$(\creatscok,\party_i,\itemid,\categid,\\\fieldid)$ from $\Fsc$, $\advS$ creates $\tx=((\creat,\allowbreak\itemid,\allowbreak\categid,\fieldid),\sgn^*,\allowbreak(\party_i,\allowbreak\pk_i^*,\allowbreak\attr_i,\allowbreak\sgn_i^*))$ and submits it to the ledger using $\Fledg$ as the real $\party_i$, where $\sgn^*$ is the signature of $(\creat,\itemid,\allowbreak\categid,\fieldid)$ with key $\sk^*_i$ using $\Fsig$.
		
		\item Upon input $(\farmiscok,\party_i,\fieldid)$ from $\Fsc$, $\advS$ simulates for $\party_i$ the creation and submission of the equivalent transaction as seen before.
		\item Upon input $(\transscok,\party_i,\itemid,\categid,\allowbreak\items)$ from $\Fsc$, $\advS$ simulates for $\party_i$ the creation and submission of the equivalent transaction as seen before.
		\item Upon input $(\mergscok,\party_i,\batchid,\products)$ from $\Fsc$, $\advS$ simulates for $\party_i$ the creation and submission of the equivalent transaction as seen before.
		\item Upon input $(\splitscok,\party_i,\batchid)$ from $\Fsc$, $\advS$ simulates for $\party_i$ the creation and submission of the equivalent transaction as seen before.
		\item Upon input $(\handsc,\p,\party_i,\party_j)$ from $\Fsc$, $\advS$ simulates for $\party_i$ the creation and submission of the equivalent transaction as seen before.
		\item Upon input $(\rejsc,\p,\party_j,\party_i)$ from $\Fsc$, $\advS$ simulates for $\party_i$ the creation and submission of the equivalent transaction as seen before.		
		\item Upon input $(\recvscok,\p,\party_i,\allowbreak\party_j)$ from $\Fsc$, $\advS$ simulates for $\party_i$ the creation and submission of the equivalent transaction as seen before.
		\item Upon input $(\audsc,\party_i,\itemid,\auddata)$ from $\Fsc$. $\advS$ checks if $\party_i$ owns a device. If $\party_i$ does not own a device $\advS$ generates a new device calling $\Fff$ with input $(\initff,\DevIss)$. Once $\Fff$ returns $(\initffok,\devid)$, $\advS$ simulates the handover from $\DevIss$ to $\party_i$ sending message $(\handff,\DevIss,\party_i)$ to $\Fff$. After that these two operations are performed, $\advS$ submits the transaction $\tx=((\audit,\itemid,\allowbreak\auddata),\sgn,(\party_i,\pk_i,\attr_i,\sgn_i))$. $\sgn$ is the signature of $(\audit,\itemid,\auddata)$ with key $\sk_i$ using $\Fsig$.		
		\item Upon input  $(\trainscok,\party_i,\allowbreak\categid,\allowbreak\cf)$\\ from $\Fsc$, $\advS$ checks if $\party_i$ owns a device. If $\party_i$ does not own a device $\advS$ generates a new device calling $\Fff$ with input $(\initff,\DevIss)$. Once $\Fff$ returns $(\initffok,\allowbreak\devid)$, $\advS$ simulates the handover from $\DevIss$ to $\party_i$ sending message $(\handff,\DevIss,\party_i)$ to $\Fff$.  $\advS$ submits the transaction $\tx=((\training,\allowbreak\categid,\cf),\allowbreak\sgn,\allowbreak(\party_i,\allowbreak\pk_i,\allowbreak\attr_i,\allowbreak\sgn_i))$. $\sgn$ is the signature of $(\training,\categid,\cf)$ with key $\sk_i$ using $\Fsig$.
	\end{itemize}

	\item Simulating corrupted regular parties. Let $\advA$ be an adversary corrupting a party $\party_i$.
	\begin{itemize}
		\item When $\advA$ instructs the party $\party_i$ to register, $\party_i$ generates $(\pk_i,\sk_i)$ invoking $\Fsig$ and sends the message $(\reg,\party_i,\pk_i,\attr_i)$ to $\advS$ via $\Fsmt$. $\advS$ writes $(\regsc,\party_i)$ to the input tape of $\Fsc$. $\advS$ generates the signature $\sgn_i$ simulating the algorithm of $\Reg$ using the private key $\sk_R^*$. $\advS$ stores $(\party_i,\pk_i,\attr_i,\sgn_i)$ and sends $(\regok,\allowbreak\party_i,\allowbreak\pk_i,\allowbreak\attr_i,\allowbreak\sgn_i)$ via $\Fsmt$ to $\party_i$.
		When $\advA$ instructs the party $\party_i$ to request an \AFS, $\party_i$ inform $\DevIss$. If $\DevIss$ has an available device (i.e., a created device not yet sold), $\DevIss$ send message $(\handff,\DevIss,\party_i)$ to $\Fff$ to inform the functionality of the handover and gives the device to $\party_i$. Otherwise, $\DevIss$ sends message $(\initff,\party_i)$ to $\Fff$ to create a new device and then performs the handover operation.
		
		\item When $\advA$ instructs the party $\party_i$ to create a new production lot, $\party_i$ sends the transaction $\tx=((\creat,\itemid,\categid,\fieldid),\sgn,(\party_i,\allowbreak\pk_i,\\\attr_i,\sgn_i))$ via $\Fledg$. $\advS$ writes $(\creatsc,\party_i,\allowbreak\itemid,\categid,\fieldid)$ to the input tape of $\Fsc$.
		
		\item When $\advA$ instructs the party $\party_i$ to create a new production area, $\party_i$ sends the transaction $\tx=((\farm,\fieldid),\sgn,(\party_i,\pk_i,\attr_i,\sgn_i))$ via $\Fledg$. $\advS$ writes $(\farmsc,\allowbreak\party_i,\fieldid)$ to the input tape of $\Fsc$.
		
		\item When $\advA$ instructs the party $\party_i$ to create a new transformed product, $\party_i$ sends the transaction $\tx=((\transform,\itemid,\items,\categid),\allowbreak\sgn,(\party_i,\pk_i,\attr_i,\sgn_i))$ via $\Fledg$. $\advS$ writes \\$(\transsc,\allowbreak\itemid,\party_i,\items,\allowbreak\categid)$ to the input tape of $\Fsc$.
		
		\item When $\advA$ instructs the party $\party_i$ to create a new training activity, $\party_i$ sends message $(\trainff,\categid)$ to $\Fff$ obtaining back the value $\cf$ associated to $\categid$. $\party_i$ sends the transaction $\tx=((\training,\categid,\cf),\sgn,(\party_i,\pk_i,\attr_i,\allowbreak\sgn_i))$ via $\Fledg$. $\advS$ writes $(\trainsc,\party_i,\allowbreak\categid,\cf)$ to the input tape of $\Fsc$.
		
		\item When $\advA$ instructs the party $\party_i$ to create a new audit activity, $\party_i$ sends message $(\evalff,\party_i,\itemid,\allowbreak\categid)$ to $\Fff$ obtaining the result of the audit. $\party_i$ sends the transaction $\tx=((\audit,\itemid,\allowbreak\auddata),\allowbreak\sgn,(\party_i,\pk_i,\attr_i,\sgn_i))$ via $\Fledg$. $\advS$ writes $(\audsc,\party_i,\itemid,\categid,\allowbreak\auddata)$ to the input tape of $\Fsc$.
		
		\item When $\advA$ instructs the party $\party_i$ to merge a set of products, $\party_i$ sends the transaction $\tx=((\merge,\allowbreak\batchid,\products),\sgn,(\party_i,\pk_i,\attr_i,\sgn_i))$ via $\Fledg$. $\advS$ writes $(\mergsc,\party_i,\batchid,\products)$ to the input tape of $\Fsc$.
		
		\item When $\advA$ instructs the party $\party_i$ to split a set of products, $\party_i$ sends the transaction $\tx=(\allowbreak(\scsplit,\batchid),\allowbreak\sgn,(\party_i,\pk_i,\attr_i,\sgn_i))$ via $\Fledg$. $\advS$ writes $(\splitsc,\party_i,\batchid)$ to the input tape of $\Fsc$.
		
		\item When $\advA$ instructs the party $\party_i$ to update a product, $\party_i$ sends the transaction $\tx=((\upd,\p,\allowbreak\newstate),\sgn,(\party_i,\pk_i,\attr_i,\sgn_i))$ via $\Fledg$. $\advS$ writes $(\updsc,\party_i,\p,\allowbreak\newstate)$ to the input tape of $\Fsc$.
		
		\item When $\advA$ instructs the party $\party_i$ to handover a product to a party $\party_j$, $\party_i$ sends the transaction $\tx=((\hand,\p,\party_i,\party_j),\sgn,(\party_i,\pk_i,\attr_i,\sgn_i ))$ via $\Fledg$. $\advS$ writes $(\handsc,\p,\party_i,\party_j)$ to the input tape of $\Fsc$.
	
		\item When $\advA$ instructs the party $\party_i$ to reject a received asset from $\party_j$, $\party_i$ sends the transaction $\tx=((\rejok,\p,\party_i,\party_j),\sgn,(\party_i,\pk_i,\attr_i,\sgn_i))$ via $\Fledg$. $\advS$ writes $(\rejsc,\p,\party_i,\party_j)$ to the input tape of $\Fsc$.

		\item When $\advA$ instructs the party $\party_i$ to accept a received asset from $\party_j$, $\party_i$ sends the transaction $\tx=((\recvok,\p,\party_i,\party_j),\sgn,(\party_i,\pk_i,\attr_i,\sgn_i))$ via $\Fledg$. $\advS$ writes $(\recvsc,\p,\party_i,\party_j)$ to the input tape of $\Fsc$.
	
		\item If $\advA$ instructs a corrupted party $\party$ to read the entire history of an asset $\assetid$, $\advS$ writes $(\readsc,\assetid)$ to the input tape of $\Fsc$.

	\end{itemize}
	\item Simulating corrupted parties interacting with $\Fff$.  Let $\advA$ be an adversary corrupting a party $\party_i$.
	\begin{itemize}
		\item When $\advA$ instructs the party $\party_i$ to perform an audit, $\party_i$ sends message $(\evalff,\party_i,\itm,\allowbreak\categid)$ to $\Fff$. \ignore{$\advS$ writes $(\audscenabl,\itemid,\party_i)$ to the input tape of $\Fsc$.} If $\party_i$ is the owner of $\itemid$, $\advS$ computes the verification algorithm of $\Fff$ and returns $(\evalffok,\allowbreak\res,\devid,\itm,\categid)$ or $(\evalffko,\allowbreak\itm,\categid)$ to $\party_i$.
		\item When $\advA$ instructs the party $\party_i$ to perform an handover of a device, $\party_i$ sends message $(\handff,\party_i,\allowbreak\party_j)$ to $\Fff$ and the handover is performed.
	\end{itemize}
\end{itemize}

We consider a sequence of hybrid experiments (ending with the real
experiment) and argue that each pair of hybrids is computationally close thanks to the properties of the underlying cryptographic primitives.

\begin{itemize}
		\item Hybrid $\HYB_3(\secpar)$: This experiment is identical to $\IDEAL_{\Fsc,\advS,\advD}(\lambda)$.
		\item Hybrid $\HYB_2(\secpar)$: $\HYB_2(\secpar)$ is identical to $\HYB_3(\secpar)$ except that $\advS$ will run the $\Reg$ algorithm to generate the signature key pair $(\pk_R,\sk_R)$ for the registration activities.
		Each party $\party_i$ that wants to use the system needs to generate a signature key pair $(\pk_{i},\sk_{i})$ using $\Fsig$. Whenever a party $\party_i$ asks for registration, $\party_i$ sends $\pk_i$ and $\attr_i$. $\advS$ signs $(\party_i,\pk_i,\attr_i)$ calling $\Fsig$ and obtains the signature $\sgn_i$ that will be used by $\party_i$ together with $\pk_i,\attr_i$ as his credential for the system. $\advS$ stores $(\party_i,\pk_i,\attr_i,\sgn_i)$ in $L$ as identifier of $\party_i$ instead of using $\party_i$.
		\item Hybrid $\HYB_1(\secpar)$: $\HYB_1(\secpar)$ is identical to $\HYB_2(\secpar)$ except that every time that a party $\party_i$ sends a command to the system, $\party_i$ will sign the command using the signature secret key $\sk_i$ and append the signature to the command together with the party identifier.
		\item Hybrid $\HYB_0(\secpar)$: This experiment is identical to $\REAL^{\Fff,\Fledg,\Fsmt,\Fsig}_{\Psc,\advA,\advD}(\secpar)$.
\end{itemize}

\begin{lemma}
	$\{\HYB_3(\secpar)\cind\HYB_2(\secpar)\}_{\secpar\in\NN}$
\end{lemma}
\begin{proof}
	We reduce to unforgeability of the signature scheme. Indeed $\advA$ can try to impersonate a different user forging the signature generated by the registration authority. In this case, $\advA$ knows only $\pk_\Reg$, and tries to generate a $\sgn_i$ for a corrupted party $\party_i$ to play in the system. If $\advD$ distinguishes $\HYB_3(\secpar)$ form $\HYB_2(\secpar)$, we will use $\advD$ to break the unforgeability of the signature scheme. Let $\advC$ be the challenger of the signature scheme and $\advA_{Sig}$ the adversary that wants to break the signature unforgeability:
	\begin{itemize}
		\item $\advC$ generates the signature key pair $(\pk_\advC,\sk_\advC)$ and sends $\pk_\advC$ to $\advA_{Sig}$;
		\item $\advA_{Sig}$ runs the experiment using as public key for the registration authority the value $\pk_\advC$;
		\item if $\advA_{Sig}$ receives a message from a not registered party $\party_i$ containing as party identifier the tuple $(\party_i,\pk_i,\attr_i,\sgn_i)$ and $(verify,\pk_\advC,(\party_i,\pk_i,\attr_i),\sgn_i)$ is true, $\advA_{Sig}$ returns to $\advC$ the values $(\party_i,\allowbreak\pk_i,\attr_i),\sgn_i$.
	\end{itemize}
This breaks the unforgeability property of the signature scheme, but the signature unforgeability can be broken only with negligible probability.

\end{proof}
\begin{lemma}
	$\{\HYB_2(\secpar)\cind\HYB_1(\secpar)\}_{\secpar\in\NN}$
\end{lemma}
\begin{proof}
	The difference between $\HYB_2(\secpar)$ and $\HYB_1(\secpar)$ is that in the transaction stored in $\HYB_2(\secpar)$ the command is signed. An adversary $\advA$ can try to send a command for which $\advA$ has no right forging the signature of another user. $\advA$ tries to send a message as $\party_i$, where $\party_i$ is a registered party, without knowing the $\party_i$ signature secret key $\sk_i$. If the transaction is accepted, it is possible to define an adversary $\advA_{Sig}$ that breaks the signature unforgeability against a challenger $\advC$:
	\begin{itemize}
		\item $\advC$ generates the signature key pair $(\pk_\advC,\sk_\advC)$ and sends $\pk_\advC$ to $\advA_{Sig}$;
		\item $\advA_{Sig}$ registers a party $\party_i$ to the experiment using message $(\reg,\party_i,\pk_\advC,\attr_i)$;
		\item $\advA_{Sig}$ can check if $\advA$ is able to produce a transaction $\tx=(\tx',\sgn,(\party_i,\pk_\advC,\attr_i))$ as user $\party_i$, where $\tx'$ is a command. In this case $\advA_{Sig}$ returns to $\advC$ the values $(\tx',\sgn)$.
	\end{itemize}
This breaks the unforgeability property of the signature scheme, but the signature unforgeability can be broken only with negligible probability.
\end{proof}

\begin{lemma}
	$\{\HYB_1(\secpar)\cind\HYB_0(\secpar)\}_{\secpar\in\NN}$
\end{lemma}
\begin{proof}
The only difference between $\HYB_1(\secpar)$ and $\HYB_0(\secpar)$ is that $\HYB_0(\secpar)$ uses $\Fledg$ to store data instead of using an internal list $L$. The simplified ledger functionality $\Fledg$ permits to every party to append bit strings to an available ledger, and every party can retrieve the current ledger. $\Fledg$ is  local functionality that guarantees that transactions are immediately appended, final, and available to all parties. Then, it is not possible to distinguish between the two different experiments.
\end{proof}
This ends the proof of Theorem~\ref{thm:sec}.
\end{proof}

\subsection{Improvements and Optimizations}\label{sec:improv}
Here we discuss issues and corresponding workarounds of our system system.
\begin{itemize}
	\item Decentralizable registration authority/ \AFS\ issuer. To keep the presentation simple, the protocol formalizes the functionality for the case of a single trusty registration authority. This trust assumption might be relaxed by distributing the registration process, which can be obtained using threshold signatures to sign the messages of the registration process, and secret sharing to share the secrets needed to register parties, etc.  We can relax the trust assumption made on the \AFS\ issuer in Section \ref{sec:threat} via decentralization. When we have several authorized \AFS\ issuers, even if an \AFS\ issuer misbehaves by creating malicious \AFS\ issuers, this behavior can be detected by other honest parties equipped with trusted \AFSs\, as analyzed in Section \ref{subsec:trust_ass}.
	\item Private data. There are cases in which a group of participants needs to store for audit purposes confidential data on-chain. This data must be kept private from other network participants. In a nutshell, a party might want to have separate private data relationships with other parties to keep private its purchasing and charging policies. For example, a wholesaler might want to hide the number of products purchased from each distributor or a manufacturer might want to hide the number of raw materials used per product. At the same time, to avoid counterfeiting, the product tracing system must guarantee that the entities cannot introduce more products than purchased passing them as the originals. Cryptographic techniques (e.g.,  homomorphic encryption, commitment schemes) can be used to guarantee both privacy preservation, public verifiability, and correctness of data in blockchains. In homomorphic encryption the encryption function has some properties that allow combining two encryptions into a third encryption of a value related to the original two, (i.e., the sum), as if such operations had been performed over the cleartext. A commitment scheme, instead, allows a party to commit to a secret value while keeping it hidden to others (hiding), with the ability to reveal the committed value later. It is infeasible for the committer to change the value after it has been committed (binding). The private data are stored off-chain by authorized organizations and only a commitment, (e.g., a hash of that data) is written to the ledgers and visible to every peer in the blockchain. The commitment serves as evidence and can be used for audit purposes. Authorized organizations might decide to share the private data via secure message channels with other parties if they get into a dispute or agreement. The parties can now compute the commitment of the private data and verify if it matches the on-chain commitment, proving that the data existed at a certain point in time.
	\item Credentials management. As explained in Section \ref{sec:entities}, during the party's registration $\Reg$ can assign different attributes as the role. These attributes can be used also to add further constraints on the operations that a party can perform according to the role. In the case of Producers or Manufacturers, they must be able to create new products only belonging to the categories that have been approved by $\Reg$ during the registration. We keep this aspect quite general. We note that further checks should be to verify if the specified location is consistent with the product certified for that user. Similar operations can also be done for other roles.
\end{itemize}

\section{Conclusions}\label{sec:concl}
In this work we have formally proposed 
and analyzed a
system for tracing physical goods leveraging blockchain
technology.
A core component of our system is the possibility of
verifying the correspondence among physical goods and 
their on-chain digital twins. 
Our system is based on a model that can be easily
adapted to accommodate additional features (e.g., secret business).

\section*{Acknowledgements}
We thank Giorgio Ciardella for several
remarks and explanations on real-world food traceability and the
effectiveness of biological fingerprint
scanning. This work has been supported in part by a scholarship funded
by Farzati Tech, in part by the European
Union's Horizon 2020 research and innovation programme under grant agreement No 780477
(project PRIViLEDGE), in part by a scholarship funded by Campania Region, and in part by the European Research Council (ERC) under agreement No 885666 (project PROCONTRA).

\bibliographystyle{splncs04}

\end{document}